\documentclass{article}
\newcommand{\note}[1]{}
\renewcommand{\note}[1]{~\\\frame{\begin{minipage}[c]{0.48\textwidth}\vspace{2pt}\center{#1}\vspace{2pt}\end{minipage}}\vspace{3pt}\\}
\newcommand{\hide}[1]{}

\usepackage{amsmath,amsthm,amssymb}

\usepackage{color}
\definecolor{Blue}{rgb}{0.9,0.3,0.3}

\newcommand{\squishlist}{
   \begin{list}{$\bullet$}
    { \setlength{\itemsep}{0pt}      \setlength{\parsep}{3pt}
      \setlength{\topsep}{3pt}       \setlength{\partopsep}{0pt}
      \setlength{\leftmargin}{1.5em} \setlength{\labelwidth}{1em}
      \setlength{\labelsep}{0.5em} } }

\newcommand{\squishlisttwo}{
   \begin{list}{$\bullet$}
    { \setlength{\itemsep}{0pt}    \setlength{\parsep}{0pt}
      \setlength{\topsep}{0pt}     \setlength{\partopsep}{0pt}
      \setlength{\leftmargin}{2em} \setlength{\labelwidth}{1.5em}
      \setlength{\labelsep}{0.5em} } }

\newcommand{\squishend}{
    \end{list}  }

\newcommand{\defeq}{:=}
\newcommand{\myvec}[1]{\mathbf{#1}}
\newcommand{\myvecsym}[1]{\boldsymbol{#1}}

\newcommand{\vbeta}{\myvecsym{\beta}}

\newcommand{\vgamma}{\myvecsym{\gamma}}

\newcommand{\vLambda}{\myvecsym{\Lambda}}

\newcommand{\vSigma}{\myvecsym{\Sigma}}

\newcommand{\vp}{\myvec{p}}

\newcommand{\vr}{\myvec{r}}

\newcommand{\vx}{\myvec{x}}

\newcommand{\vy}{\myvec{y}}

\newcommand{\vX}{\myvec{X}}

\newcommand{\vY}{\myvec{Y}}

\newcommand{\diag}{\mathrm{diag}}

\newcommand{\calD}{{\cal D}}

\newcommand{\data}{\calD}

\newcommand{\be}{\begin{equation}}
\newcommand{\ee}{\end{equation}}
\newcommand{\bea}{\begin{eqnarray}}
\newcommand{\eea}{\end{eqnarray}}
\newcommand{\beaa}{\begin{eqnarray*}}
\newcommand{\eeaa}{\end{eqnarray*}}

\DeclareMathAlphabet{\mathpzc}{OT1}{pzc}{m}{n}

\newtheorem{mydefinition}{Definition}
\newtheorem{proposition}[mydefinition]{Proposition}

\usepackage{graphicx} 
\usepackage{subfigure} 
\usepackage{natbib}
\usepackage{algorithm}
\usepackage{algorithmic}
\graphicspath{{figures/}}
\usepackage{flushend} 

\usepackage{hyperref}

\renewenvironment{thebibliography}[1]{%
	\begin{oldthebibliography}{#1}%
		\setlength{\parskip}{0.9ex}%
		\setlength{\itemsep}{0ex}%
	}%
	{%
	\end{oldthebibliography}%
}
\renewcommand{\refname}{\textbf{\normalsize References}}
\def\bibfont{\small}

\usepackage[accepted]{icml2013}

\icmltitlerunning{Adaptive Hamiltonian and Riemann Manifold Monte Carlo Samplers}

\begin{document}
\twocolumn[
\icmltitle{Adaptive Hamiltonian and Riemann Manifold Monte Carlo Samplers}

\icmlauthor{Ziyu Wang}{ziyuw@cs.ubc.ca}
\icmladdress{University of British Columbia
            Vancouver, Canada}
\icmlauthor{Shakir Mohamed}{shakirm@cs.ubc.ca}
\icmladdress{University of British Columbia
            Vancouver, Canada}
            \icmlauthor{Nando de Freitas}{nando@cs.ubc.ca}
\icmladdress{University of British Columbia
            Vancouver, Canada}
            
\icmlkeywords{Hybrid Monte Carlo, Hamiltonian Monte Carlo, adaptive MCMC, bandits, Bayesian optimization}

\vskip 0.3in]

\begin{abstract}
In this paper we address the widely-experienced difficulty in tuning Hamiltonian-based Monte Carlo samplers. We develop an algorithm that allows for the adaptation of Hamiltonian and Riemann manifold Hamiltonian Monte Carlo samplers using Bayesian optimization that allows for infinite adaptation of the parameters of these samplers. We show that the resulting sampling algorithms are ergodic, and that the use of our adaptive algorithms makes it easy to obtain more efficient samplers, in some cases precluding the need for more complex solutions. Hamiltonian-based Monte Carlo samplers are widely known to be an excellent choice of MCMC method, and we aim with this paper to remove a key obstacle towards the more widespread use of these samplers in practice.
\end{abstract}

\section{Introduction}
\label{sec:introduction}
Hamiltonian Monte Carlo (HMC) \cite{Duane-87} is widely-known as a powerful and efficient sampling algorithm, having been demonstrated to outperform many existing MCMC algorithms, especially in problems with high-dimensional, continuous, and correlated distributions \cite{chen-01,Neal-10}.  Despite this flexibility, HMC has not been widely adopted in practice, due principally to the sensitivity and difficulty of tuning its hyperparameters. In fact, tuning HMC has been reported by many experts to be more difficult than tuning other MCMC methods \cite{Iswharan-99, Neal-10}. In this paper we aim to remove this obstacle in the use of  HMC by providing an automated method of determining these tunable parameters, paving the way for a more widespread application of HMC in statistics and machine learning.
\\ \\
There are few existing works dealing with the automated tuning of HMC. Two notable approaches are: the No U-turn sampler (NUTS) \citep{hoffman2011no}, which is an adaptive algorithm for HMC that aims to find the best parameter settings by tracking the sample path and preventing HMC from retracing its steps in this path; and Riemann manifold HMC (RMHMC) \citep{Girolami-11}, which provides adaptations using the Riemannian geometry of the problem.

In this paper, we follow the approach of \textit{adapting} Markov chains in order to 
improve the convergence of both HMC and RMHMC. Our adaptive strategy is based on Bayesian optimization; see for example \citet{Brochu-09} and \citet{Snoek-12} for a clear and comprehensive introduction to Bayesian optimization. 
Bayesian optimization has been proposed previously  for the adaptation of general MCMC samplers by~\citet{Mahendran-12} and \citet{Hamze-11}. 
To guarantee convergence, these works were limited to a finite adaptation of the Markov chain. However, in the field of adaptive MCMC, it is well known 
that finite adaptation can result in the sampler being trapped in suboptimal parameter settings, leading to inefficient sampling. 
\\ \\
We describe Hamiltonian-based Monte Carlo samplers in section \ref{sec:hmc}, and then make the following contributions: 
\begin{itemize}
	\setlength{\itemsep}{1pt}
	\setlength{\parskip}{1pt}
	\setlength{\parsep}{0pt}
	\item We present an algorithm for adaptive HMC in which we allow for infinite adaptation of the Markov chain, thus avoiding parameter traps due to finite adaptation (section \ref{sec:ahmc}).
	\item Importantly, we prove that the adaptive MCMC samplers we present are ergodic in this infinite adaptation setting (section \ref{sec:convergence}). 
	\item We provide a comprehensive set of experiments demonstrating that the adaptive schemes perform better in a diverse set of statistical problems (section \ref{sec:results}). 
	\item For most examples, we use a version of the \emph{expected squared jumping distance} proposed by \citet{Pasarica-10} as the objective function for adaptation. However, in section~\ref{sec:results}, we also introduce a new approach for adaptive MCMC based on predictive measures, for use in settings where it is possible to perform cross-validation or bootstrapping.
\end{itemize}

\section{Hamiltonian-based Monte Carlo Sampling}
\label{sec:hmc}
Hamiltonian (or Hybrid) Monte Carlo \citep{Duane-87, Neal-10}, has become established as a powerful, general purpose Markov chain Monte Carlo (MCMC) algorithm for sampling from general, continuous distributions. Its efficiency is due to the fact that it makes use of gradient information from the target density to allow for an ergodic Markov chain capable of large transitions that are accepted with high probability. 
This efficiency and flexibility is demonstrated by the wide range of applications to which HMC has been applied, including:
Bayesian generalized linear models \citep{Iswharan-99}, Bayesian neural networks \citep{Neal-06}, Gaussian process regression and classification \citep{Rasmussen-06}, exponential family PCA and factor analysis \citep{mohamed_BXPCA}, and restricted Boltzmann machines \citep{Ranzato-10a}, amongst others.

For HMC, we are required to specify a potential energy function, which is the log of the joint distribution we wish to sample from, ${U}({\bf x}) = -\log p({\bf x})$ and a kinetic energy function, most typically, ${K}({\bf p}) = {\bf p}^{T}{\bf M}^{-1}{\bf p}/2$, with momentum vector $\bf p$ and a positive definite mass matrix ${\bf M}$. For standard HMC, the mass matrix is set to the identity. We defer the technical details of HMC to existing work \citep{Neal-10}, and present only the algorithm here (Alg. \ref{alg:hmc}).

HMC requires the selection of two free parameters: a step-size $\epsilon$ and a number leapfrog steps $L$. The accepted guidance is to choose a step-size to ensure that the sampler's rejection rate is between 25\%-35\%. It is also preferable to have a large $L$, since this reduces the random walk behavior of the sampler \citep{Neal-10}, but too large an $L$ results in unnecessary computation. 
In this paper, we consider a slight variation of the HMC algorithm: 
instead of performing $L$ leapfrog steps at each iteration, we 
only perform a random number of leapfrog steps, chosen from the discrete uniform distribution over $\{1,\cdots, L\}$, i.e. $L_r \sim \mathcal{U}(1, L)$ steps. This approach amounts to using a mixture of $L$ different HMC transition kernels, thus preserving detailed balance~\cite{Andrieu-03}.

\begin{algorithm}[t]
\caption{Hamiltonian Monte Carlo Algorithm}
\begin{algorithmic}[1]
\label{alg:hmc}
{
\small
\STATE Given: $M$, $L$, $\epsilon$, and $\vx^1$.
\FOR{$t=1,2,\cdots$}
  \STATE Sample $\vp^{t} \sim \mathcal{N}({\bf 0}, M)$ and $L_r \sim \mathcal{U}(1, L)$
  \STATE Let $\vx_0 = \vx^t$ and $\vp_0 = \vp^t+\frac{\epsilon}{2}\left. \frac{\partial U}{\partial {\bf x}}\right|_{{\bf x}_{0}}$
  \FOR{$l=1,2,\cdots, L_r $}
    \STATE ${\bf x}_{l} = {\bf x}_{l-1} + \epsilon M^{-1}{\bf p}_{l-1}$
    \STATE $\vp_l = \vp_{l-1}+\epsilon \left.\frac{\partial U}{\partial {\bf x}}\right|_{{\bf x}_{l}}$
  \ENDFOR
  \STATE $\vp_l = \vp_{l-1}-\frac{\epsilon}{2} \left.\frac{\partial U}{\partial {\bf x}}\right|_{{\bf x}_{l}}$
  \STATE Draw ${\bf u} \sim \mathcal{U}(0, 1)$
  \IF{${\bf u} < \min[1, e^{U({\bf x}^{t}) + K({\bf p}^{t}) - U({\bf x}_{l})-K({\bf p}_{l})}]$} 
    \STATE Let $({\bf x}^{t+1},{\bf p}^{t+1}) = ({\bf x}_{l},{\bf p}_{l})$ 
  \ELSE 
    \STATE Let $({\bf x}^{t+1},{\bf p}^{t+1}) = ({\bf x}^{t},{\bf p}^{t})$
  \ENDIF
\ENDFOR
}
\end{algorithmic}
\end{algorithm}
HMC is known to be highly sensitive to the choice of $\epsilon$ and $L$. 
We demonstrate HMC's sensitivity to these parameters by sampling from a bivariate Gaussian with correlation coefficient 0.99. We consider three settings ($\epsilon, L) = \{(0.16, 40), (0.16,50), (0.15,50) \}$ and show the behavior of the sampler as well as the autocorrelation plot in figure \ref{fig:traj}. While the first setting exhibits good behavior and low auto-correlation, small changes to these settings results in poor mixing and high auto-correlation, as seen in the other graphs. Theoretical results concerning the 
optimal acceptance rate for HMC exist, having been described by \citet{Beskos-10} and \citet{ Neal-10}, with
both concluding a rate around $0.65$ as optimal. Such results, however, would not help in choosing the best sampler out of 
the three in Figure~\ref{fig:traj}, since all three samplers in this demonstration have an acceptance rate around $0.7$, leaving little guidance for finding the most efficient sampler.

\begin{figure}[t]
\centering
\includegraphics[height = 4.3cm]{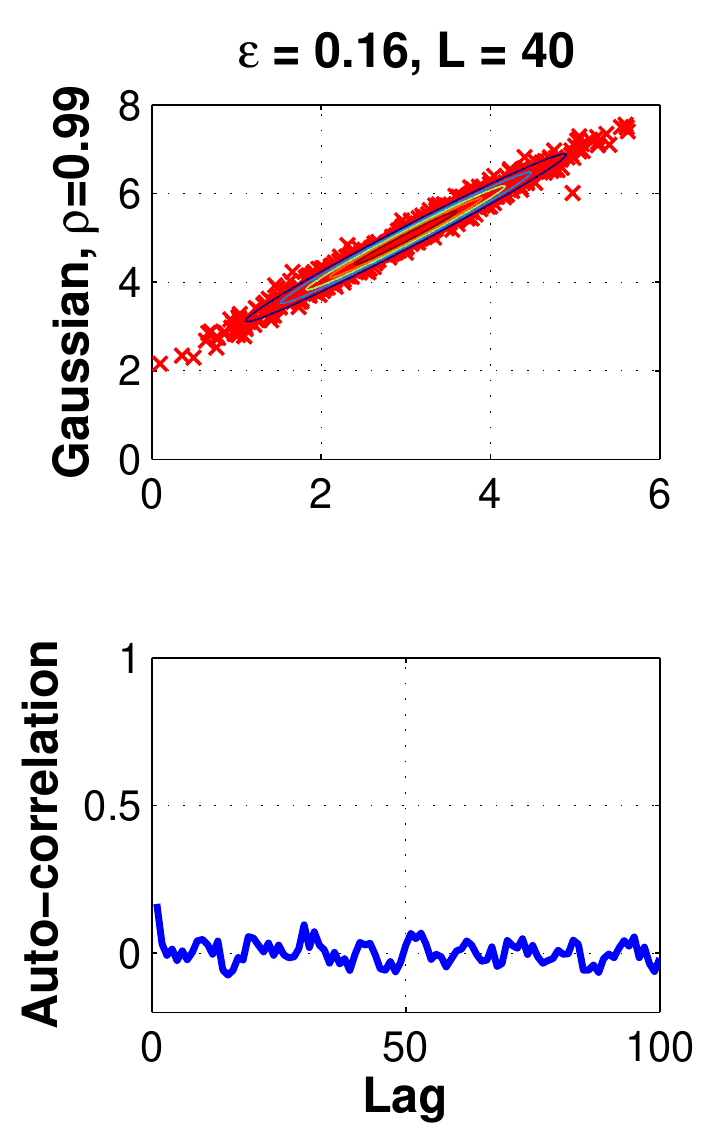}
\includegraphics[height = 4.3cm]{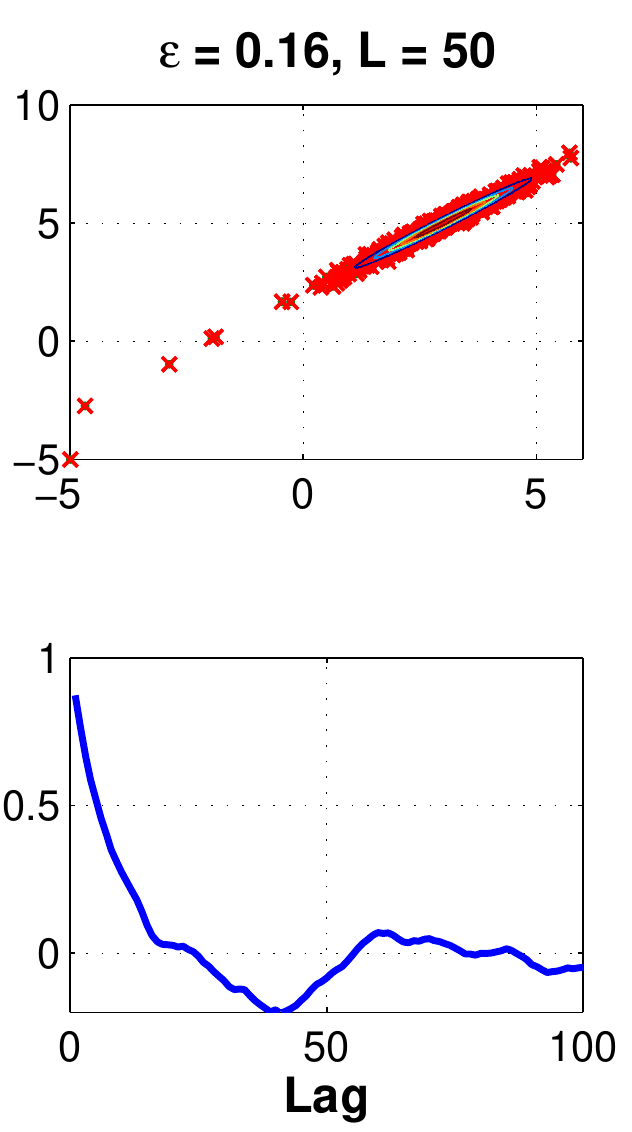}
\includegraphics[height = 4.3cm]{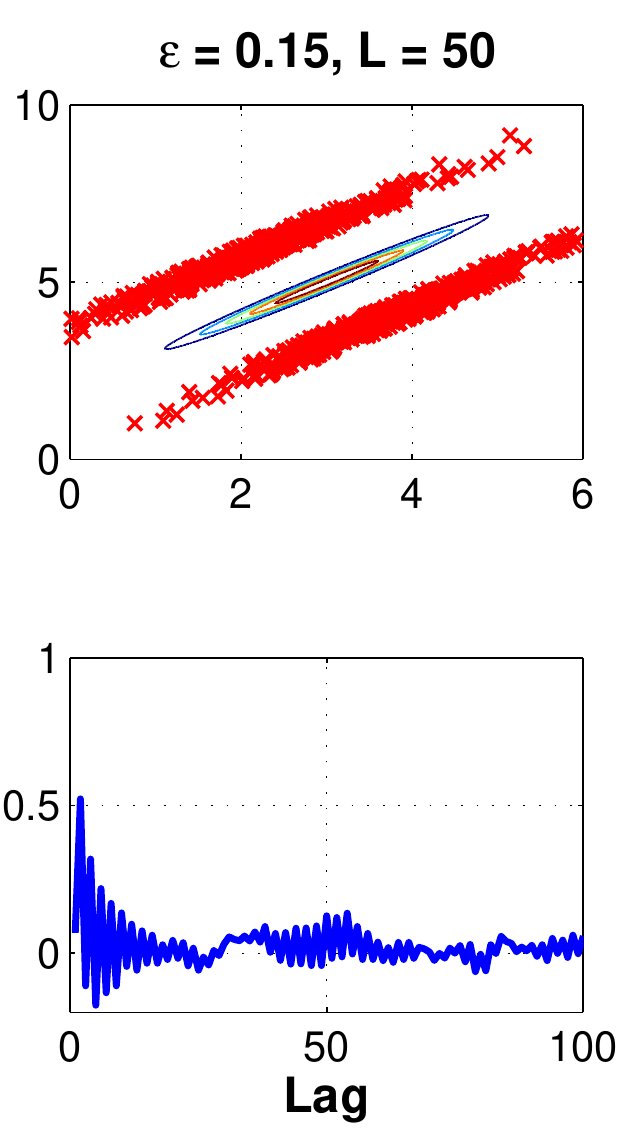}
\caption{1000 samples from a bivariate Gaussian distribution generated using HMC. We show the trajectory and auto-correlation of the samples for 3 parameter settings. 
}
  \label{fig:traj}
  \vspace{-0.5cm}
\end{figure}

To address the problem of choosing these parameters, we will introduce a method for automatically and adaptively tuning the parameters of HMC, reducing the need for time-consuming, expert tuning. An existing approach for automatic tuning of HMC was introduced by \citet{hoffman2011no}, referred to as the No U-turn sampler (NUTS). NUTS allows for automatic tuning of both HMC's parameters by tuning the stepsize $\epsilon$ during the burn-in phase, after which it is fixed and the number of leapfrog steps is adjusted thereafter for every iteration. $\epsilon$ is chosen using a stochastic approximation method referred to as dual averaging, and $L$ is chosen for every sample using a recursive algorithm in which the number of leapfrog steps is allowed to increase until the proposal trajectory taken by the sampler begins to move back towards the initial point, thus preventing U-turns and allowing for the good mixing of the chain. 

Riemann manifold HMC (RMHMC) \citep{Girolami-11} is a sampling method derived from HMC, and provides an adaptation mechanism for HMC by exploiting the Riemannian geometry of the parameter space. Rather than adapting $\epsilon$ and $L$, RMHMC accounts for the local structure of the joint density by adapting the mass matrix $\bf M$ used in HMC. Since RMHMC automatically adapts its mass matrix, the stepsize $\epsilon$ is usually fixed and the number of leapfrog steps $L$, which is a single scalar, can be chosen using the rejection rate. While the sensitivity to these parameters is greatly reduced, they must still be set and there is no general guidance on how these parameters should be chosen,  making it desirable to have a fully automatic method for RMHMC as well.

\section{Adaptive Hamiltonian Monte Carlo}
\label{sec:ahmc}

In order to adapt the MCMC parameters $L$ and 
$\epsilon$ for HMC, we need to (i) define an objective function and (ii) choose a suitable optimization method.

As pointed out in \citet{Pasarica-10}, a natural objective function for adaptation is the asymptotic efficiency of an MCMC sampler, 
$(1+2\sum_{k=1}^{\infty}\rho_k)^{-1}$,
where $\rho_k$ is the auto-correlation of the sampler with lag $k$.
Despite its appeal, this measure is problematic because the higher
order auto-correlations are hard to estimate. 
To circumvent this problem, Pasarica and Gelman (2010) introduced an objective measure 
called the expected squared jumping distance (ESJD):
$$\mbox{ESJD}(\vgamma) = \mathbb{E}_{\vgamma} \|\vx^{t+1} - \vx^{t}\|^2,$$
where $\vgamma =(L,\epsilon)$ denotes the set of parameters for HMC.
Maximizing the above objective is equivalent to
minimizing the first-order auto-correlation $\rho_1$. In practice, the above intractable expectation with respect to the Markov chain is approximated by an empirical estimator, as outlined in \citet{Pasarica-10}. 

The ESJD measure is efficient in situations where the higher order auto-correlations
increase monotonically with respect to $\rho_1$. However, it is not suitable for tuning HMC samplers since by increasing the number of leapfrog
steps one can almost always generate better samples. What we need is a measure that also takes computing time into consideration. 
 With this goal in mind, we introduce the following
objective function:
$$f(\gamma)=
\frac{\mbox{ESJD}(\vgamma)}{\sqrt{L}} 
= \frac{\mathbb{E}_{\vgamma} \|\vx^{t+1} - \vx^{t}\|^2}{\sqrt{L}}.$$
This function simply normalizes the ESJD by the number of leapfrog steps $L$, thus taking both statistical efficiency and computation into consideration. Most of our experiments will use this measure as we have found it to work very well in practice. Many works in the adaptive MCMC literature have considered matching empirical and theoretical acceptance rates in order to adapt MCMC samplers; see for example \citet{Andrieu-01} or \citet{Vihola-10}. We have found this strategy to perform poorly in the case of HMC, where samplers with the same acceptance rate can exhibit different mixing behavior (figure \ref{fig:traj}). When discussing Bayesian neural networks in our experiments (section \ref{sect:expt_BNN}), we will introduce an alternative objective function based on predictive performance. Such a measure does however only apply in predictive domains and is, consequently, less general than the normalized ESJD objective.

Now that we are armed with an objective function, we need to address the issue of optimization. Since the objective is only available point-wise (that is, it can be evaluated but its exact form is intractable), researchers typically use stochastic approximation. We use Bayesian optimization to optimize the objective. A discussion contrasting these two alternatives is presented in \citet{Hamze-11}. 

Bayesian optimization is an efficient 
gradient-free optimization tool well suited for expensive black box functions. Our objective function (normalized ESJD) is of this nature.
As mentioned earlier, normalized ESJD involves an intractable expectation that can be approximated by sample averages, where the samples are produced by running HMC for a few iterations. Each set of HMC samples for a specific set of hyper-parameters $\vgamma \in \Gamma$ results in a noisy evaluation of the normalized ESJD: $r(\vgamma) = f(\vgamma) + \varepsilon$, where we assume that the measurement noise is Gaussian $\varepsilon \sim \mathcal{N}(0, \sigma_{\eta}^2)$.

Following the standard Bayesian optimization methodology, we set $\Gamma$ to be a box constraint such that
$$\Gamma = \{(\epsilon, L):  \epsilon 
\in [b_l^{\epsilon}, b_u^{\epsilon}], L \in [b_l^{L}, b_u^{L}] \}$$ for some interval boundaries
$b_l^{\epsilon} \leq b_u^{\epsilon}$ and $b_l^{L} \leq b_u^{L}$. The parameter $L$ is discrete. The parameter $\epsilon$ is continuous, but since it is one-dimensional, we can discretize it using a very fine grid.

Since the true objective function is unknown, we specify a zero-mean Gaussian prior over it:
$$
f(\cdot) \sim GP(0, k(\cdot,\cdot))
$$
where $k(\cdot, \cdot)$ is the covariance 
function. 
Given noisy evaluations of the objective function 
$\{r_k\}_{k=1}^{i}$ evaluated at points 
$\{\vgamma_k\}_{k=1}^{i}$, we form the dataset
$\data_i = \left(\{\vgamma_k\}_{k=1}^{i}, \{\vr_k\}_{k=1}^{i} \right)$. 
Using Bayes rule, we arrive at the posterior 
predictive distribution over the unknown objective function:
\begin{align*}
f | \mathcal{D}_i,\vgamma &\sim
\mathcal{N}
(\mu_i(\vgamma), \sigma^2_i(\vgamma)) \\
\mu_i(\vgamma) &=
\mathbf{k}^{T}(\mathbf{K + \sigma_{\eta}^2 I})^{-1}\mathbf{r}_i \\
\sigma^2_i(\vgamma) &=
k(\vgamma, \vgamma) -
\mathbf{k}^{T}(\mathbf{K + \sigma_{\eta}^2 I})^{-1}\mathbf{k}
\end{align*}
where
\begin{align*}
\mathbf{K} &=
\begin{bmatrix}
k(\vgamma_1, \vgamma_1) & \ldots & k(\vgamma_1, \vgamma_i) \\
\vdots & \ddots & \vdots \\
k(\vgamma_i, \vgamma_1) & \ldots & k(\vgamma_i, \vgamma_i)
\end{bmatrix},
\end{align*}
$\mathbf{k} =
[
k(\vgamma, \vgamma_1) \; \ldots \;
k(\vgamma, \vgamma_i)]^T,$ and $\vr_i = [
r_1 \; \ldots \; r_i]^T.$ In this work, we adopt a Gaussian ARD
covariance function with $k(\vgamma_i, \vgamma_j) = \exp(-\frac{1}{2}\vgamma_i^T \Sigma^{-1} \vgamma_j)$
where $\Sigma$ is a positive definite matrix.
We set
$
\Sigma = \mathrm{\diag}\left( \left[\alpha(b_u^{\epsilon}-b_l^{\epsilon})\right]^2; \left[\alpha(b_u^{L}-b_l^{L})\right]^2 \right),
$ where $\alpha = 0.2$.
For more details on Gaussian processes, please refer to \citet{Rasmussen-06}.

\begin{algorithm}[t!]
\caption{Adaptive HMC.}
\label{alg:adaptive-mcmc}
\begin{algorithmic}[1]
{
\small
  \STATE Given: $\Gamma$, $m$, $k$, $\alpha$, and $\vgamma_1$.
  \FOR{$i=1,2,\dots, $}
    \STATE Run HMC for $m$ iterations with $\vgamma_{i}=(\epsilon_i,L_i)$.
    \STATE Obtain the objective function value $r_{i}$ using the drawn samples.
    \STATE Augment the data $\mathcal{D}_{i} = \{\mathcal{D}_{i-1}, (\vgamma_i,r_i )\}.$  
    \IF{$r_{i} > \sup_{j \in \{1, \cdots, i-1\}}r_j$}
	\STATE $s = \frac{\alpha}{r_{i}}$
    \ENDIF
    \STATE Draw ${ u} \sim \mathcal{U}([0, 1])$
    \STATE let $p_i = (\max\{i-k+1, 1\})^{-0.5} $, with $k \in \mathbb{N^+}$.
    \IF{${ u} < p_i$}
      \STATE $\vgamma_{i+1} \defeq 
      \arg\max_{\vgamma \in \Gamma} {u}(\vgamma,s|\mathcal{D}_{i})$.
    \ELSE 
      \STATE $\vgamma_{i+1} \defeq \vgamma_{i}$
    \ENDIF
  \ENDFOR
}
\end{algorithmic}
\end{algorithm}

The Gaussian process simply provides a surrogate model for the true objective. The surrogate can be used to search, efficiently, for the maximum of the objective function. In particular, it enables us to construct an acquisition function ${u}(\cdot)$ that tells us which parameters $\vgamma$ to try next. The acquisition function uses the Gaussian process posterior mean to predict regions of potentially higher objective values (exploitation). It also uses the posterior variance to detect regions of high uncertainty (exploration). Moreover, it effectively trades-off exploration and exploitation. Different acquisition functions have been proposed in the literature \cite{Mockus-82,Srinivas-10,Hoffman-11}. We adopt a variant of the Upper Confidence Bound (UCB) \cite{Srinivas-10}, modified to suit our application:
$${u}(\vgamma,s|\data_i) = \mu_i(\vgamma,s) + p_i\beta_{i+1}^{\frac{1}{2}}\sigma_i(\vgamma).$$
As in standard UCB, we set $\beta_{i+1} = 2\log\left( \frac{(i+1)^{\frac{d}{2}+2}\pi^2}{3\delta}\right)$,
where $d$ is the dimension of $\Gamma$ and $\delta$ is set to $0.1$. The parameter $p_i$ ensures that the diminishing adaptation condition for adaptive MCMC \cite{Roberts-07} is satisfied. Specifically, we set
$p_i = (\max\{i-k+1, 1\})^{-0.5} $ 
for some $k \in \mathbb{N^+}$. As $p_i$ goes to $0$, the probability of Bayesian optimization adapting $\vgamma$ vanishes as shown in Algorithm \ref{alg:adaptive-mcmc}.

It could be argued that this acquisition function could lead to premature
exploitation, which may prevent Bayesian optimization from locating the true optimum of the
objective function. There is some truth to this argument. Our goal when adapting the Markov
chain, however, is less about finding the absolute best hyper-parameters but more 
about finding sufficiently good hyper-parameters given finite computational resources. 
Given enough time, we could slow the annealing schedule 
thus allowing Bayesian optimization to explore the hyper-parameter space fully. However, under time constraints we must use faster annealing schedules. As $p_i$ decreases, it becomes increasingly difficult for Bayesian optimization to propose
new hyper-parameters for HMC. Consequently, the sampler ends up using the same
set of hyper-parameters for many iterations. With this in mind, we argue, it is more 
reasonable to exploit known good hyper-parameters rather than exploring for 
better ones. This intuition matches our experience when conducting experiments.

The acquisition function also includes a scalar scale-invariance parameter $s$, such that $\mu_i(\vgamma,s) =
\mathbf{k}^{T}(\mathbf{K + \sigma_{\eta}^2 I})^{-1}\mathbf{r}_i s$. This parameter is estimated automatically so as to rescale the rewards to the same range each time we encounter a new maximal reward.

Gaussian processes require the inversion of the covariance 
matrix and, hence, have complexity $\mathcal{O}(i^3)$, where $i$ is the number of iterations. 
Fortunately, thanks to our annealing schedule, the number
of unique points in our Gaussian process grows sub-linearly
with the number of iterations. This slow growth makes it possible to adopt kernel specification techniques, as proposed by  \citet{Engel-05}, 
to drastically reduce the computational cost without suffering any loss in accuracy.

Finally, in all our experiments, we set $\alpha = 4$, $k=100$, $m = \frac{B}{k}$, 
where $B$ is the number of burn-in samples. In our experience,
the algorithm is robust with respect to these settings and we used the same set
of parameters throughout our experiments with the exception of $\Gamma$. 
$\Gamma$ is easy to set, since one can choose the bound to be large enough
to contain all reasonable $\epsilon$ and $L$, while allowing the adaptive
algorithm enough time to explore. Alternatively, one could gauge the hardness of the
sampling problem at hand and set more reasonable bounds. In general, harder sampling problems require a smaller $\epsilon$ and a larger $L$. We follow this second strategy 
throughout our experiments and found that most sensible bounds led to performance similar to the ones reported.

\section{Analysis of Convergence}
\label{sec:convergence}

The proof of ergodicity of the adaptive HMC algorithm capitalizes on existing results for Langevin diffusions and adaptive MCMC on compact state spaces. The method of proof is based on the standard Lyapunov stability functions, also known as drift or potential functions. 

We assume that our target distribution is compactly supported on $\mathcal{M}$. 
In practice, for target distributions that are not compactly supported,
we could set $\mathcal{M}$ large enough to contain most of the mass of our target distribution.
The sampler is restricted to $\mathcal{M}$ by following this standard approach of rejecting all proposals that
fall outside $\mathcal{M}$. 

Let $\{P_{\gamma}\}_{\gamma \in \Gamma}$ be a collection of Markov chain kernels, 
each admitting $\pi$ as the stationary distribution.
That is, for each value of $\gamma = (\epsilon, L)$, we have one such kernel. Moreover, let $P_{\gamma}^n$ denote the $n$-step Markov kernel.
Our proof requires the following classical definitions:
\begin{mydefinition}
  {\upshape ({\bf Small set})} A subset of the state space $C \subseteq \mathcal{X}$ is small if there 
  exists $n_0 \in \mathbb{N}^+$, $\xi > 0$ and a probability measure $\nu(.)$
  such that ${P}^{n_0}(x, \cdot) \geq \xi \nu(\cdot)$ $\forall x \in C$.
\end{mydefinition}
\begin{mydefinition}
  {\upshape ({\bf Drift condition})} A Markov chain satisfies the drift condition if 
  for a small set $C$, there exist constants $ 0 < \lambda < 1$ and $b < \infty$,
  and a function $V: \mathcal{X} \rightarrow [1, \infty]$ such that $\forall x \in \mathcal{X}$
  $$\int_{\mathcal{X}}P(x, dy)V(y) \leq \lambda V(x) + b {\bf 1}_{C}(x).$$ 
\end{mydefinition}
Having defined the necessary concepts, we now move on to show the ergodicity 
of our adapted approach.
\begin{proposition}
  Suppose that $P_{\gamma}$, when restricted to a compact set $\mathcal{M}$, admits
  the stationary distribution $\pi$ for all $\gamma \in \Gamma$.
  If $\pi$ is continuous, positive and bounded on $\mathcal{M}$, and $|\Gamma|$ is finite, 
  then the adaptive HMC sampler is ergodic. 
\end{proposition}
\begin{proof}
  To show that adaptive HMC converges on a compact set, we first show that 
  $\mathcal{M}$ is a small set. 
  
  The transition kernel of the random time HMC algorithm can be written as 
  ${P}_{\gamma}(x, .) = \sum_{l=1}^{L} \frac{1}{L}{Q}_{l,\epsilon}(x, .)$ 
  where ${Q}_{l,\epsilon}(x, .)$ is the transition kernel of an HMC sampler that takes $l$
  leapfrog steps with parameter $\epsilon$. In particular ${Q}_{1,\epsilon}(x, .)$ is the transition kernel of
  Metropolis adjusted Langevin algorithm (MALA). 
  As $\pi$ is bounded, and the proposal distribution of MALA is positive every where,
  we have that ${Q}_{1,\epsilon}$ is $\mu^{Leb}$-irreducible.
  By a slight modification of Theorem 2-2 in \citet{Roberts-96}, for Markov chains defined by MALA,
  and any compact set $C$ with positive Lebesgue measure 
  (i.e. $\mu^{Leb}(C)>0$) there exists $\xi > 0$ and a probability measure
  $\nu(\cdot)$ such that $\forall x \in C$
  ${Q}_{1,\epsilon}^{1}(x, .) \geq \xi \nu(.).$ Hence, $\mathcal{M}$ is a small set since
  $${P}_{\gamma}^{1}(x, .) \geq \frac{1}{L}{Q}_{1,\epsilon}^{1}(x, .) \geq \frac{1}{L} \xi \nu(.)$$
  for any compact set $C$ where $\mu^{Leb}(C)>0$. 
  The drift condition is 
  trivially satisfied by each HMC sampler when we choose $C$ to be $\mathcal{M}$, and $V$ to 
  be such that $V(x) = 1$ for all $x$.

  Having proved these conditions, we can now appeal to 
  Theorem 15.0.1 of~\citet{meyn-93} to conclude that 
  $\|P_{\gamma}^n(x, \cdot) - \pi(\cdot) \| < R_{\gamma} V(x) \rho^{n}_{\gamma}$ for all $n$ and
  for $0<\rho_{\gamma}<1$. 
  Since $V(X) = 1$ $\forall x$, we have 
  $$\|P_{\gamma}^n(x, \cdot) - \pi(\cdot) \| < R_{\gamma} \rho^{n}_{\gamma}.$$
  Define $R_{max} = \sup_{\gamma \in \Gamma} R_{\gamma}$ and 
  $\rho_{max} = \sup_{\gamma \in \Gamma} \rho_{\gamma}$, then 
  $\forall x \in \mathcal{M}$ and $\forall \gamma \in \Gamma$ we have
  $$\|P_{\gamma}^n(x, \cdot) - \pi(\cdot) \| < R_{max} \rho^{n}_{max}.$$

  We have shown that the kernels $\{P_{\gamma}(x, \cdot)\}_{\gamma \in \Gamma}$ are simultaneously uniformly ergodic. 
  Also, the adaptive HMC sampler has diminishing adaptation by design.  
  By Theorem 5 of \citet{Roberts-07}, these two conditions imply the claim of our proposition.
\end{proof}

In general two sets of conditions together guarantee ergodicity of
an adaptive MCMC algorithm~\cite{Roberts-07, atchadé-10}. 
First, the adaptation has to diminish eventually.
The second set of conditions is usually placed on the underlying 
MCMC samplers.
In \citet{Roberts-07}, the samplers are required to be either
simultaneously uniformly or geometrically ergodic. 
Without restricting the state space to be compact, 
it is unlikely that HMC is uniformly ergodic. Also, to the best of our 
knowledge, no theoretical results exist on the geometric ergodicity of HMC
when the state space is not compact. 
However, \citet{Roberts-02} showed that Langevin diffusion, 
which is closely related to HMC, is geometrically ergodic. 
Thus one potential challenge would be to prove or disprove
geometric ergodicity of HMC in general state spaces. 
\citet{atchadé-10} weakened 
the conditions required, still requiring diminishing adaptation, but the requirements on 
the underlying MCMC samplers were reduced to sub-geometric ergodicity. 
Although these conditions are weaker, it remains hard to check whether 
HMC satisfies them.

\section{Results}
\label{sec:results}
We show the performance of our adaptive algorithm on four widely-used models. We evaluate the performance of the samplers using the effective sample size (ESS) using: $ESS = R \left( 1 + 2 \sum_k \rho_k \right)$, where $R$ is the number of posterior samples, and $\sum_k \rho_k$ is the sum of $K$ monotone sample auto-correlations computed using the monotone sequence estimator \citep{Girolami-11}. We adopt the total number of leapfrog steps used in producing the set of samples as a proxy for computational demand, since the computation is dominated by the gradient evaluation required for each leapfrog step. An efficient sampler will result in the highest ESS for the least computation, and we will thus report the effective sample size per leapfrog step used (ESS/L), similarly to \citet{hoffman2011no}, since this takes into account computational requirements. We compute the ESS/L over all dimensions of the target distribution and report the minimum, median and maximum ESS obtained. While we report  all three summary statistics, we focus on the \emph{minimum ESS/L} as the most useful measure, since this allows us to evaluate the efficiency of the most confined coordinate, and is more indicative of ESS jointly over all coordinates rather than, as computed, over every coordinate independently \citep{Neal-10, Girolami-11}.

We compare our adaptive HMC to NUTS, and  extend our approach and compare an adaptive version of RMHMC to the standard RMHMC.
For NUTS, we tuned the free parameters of its dual averaging algorithm to obtain the best performance, and for RMHMC we use the experimental protocol and code used by \citet{Girolami-11}. We do this for all experiments in this section. Code to reproduce these results will be available online. 

\subsection{Bayesian Logistic Regression}
We consider a data set $\vX$ consisting of $N$ observations and $D$ features or covariates, and a binary label $\vy$. Using regression coefficients $\vbeta$ and bias $\beta_0$ the joint distribution for the logistic regression model is:
\begin{align}
&\!\!\log p(\vX,\! \vy, \!\vbeta, \!\beta_0) \!\propto\! \log p(\vy | \vX, \!\vbeta, \!\beta_0) \!+\! \log p(\vbeta) \!+\! \log p(\beta_0) \nonumber \\
& = \!-\! \sum_i \!\log\left (\!1 \!+\!\exp \left(\! -y_i (\beta_0 \!+\! \vx_i^\top \! \vbeta)\! \right) \! \right) \!-\! \frac{\beta_0^2}{2 \sigma^2}  \!-\! \frac{\vbeta^\top\! \vbeta}{2 \sigma^2},
\end{align}
where $y_i \in \{-1, 1\}$, and $\sigma^2$ is the prior variance of the regression coefficients. We present results on five data sets from the UCI repository. The data sets have varying characteristics with features $D$ ranging from 2 to 24, and the number of observations from 250 to 1000. For each data set, we generate 5000 samples after a burnin phase of 1000 samples, and repeat this process 10 times using differing starting points. 
The top row of figure \ref{fig:ESSL_BLR} compares the performance of our adaptive HMC (AHMC) to NUTS, while the bottom row compares our adaptive RMHMC (ARMHMC) to RMHMC. 
For this experiment, for AHMC, we set $\Gamma$ such that $\epsilon \in [0.01, 0.2]$ and $L\in\{0, \cdots, 100\}$, and for ARMHMC, we use $\epsilon \in [0.1, 1]$ and $L\in\{1, \cdots, 12\}$.

The columns of figure \ref{fig:ESSL_BLR} shows box plots of the minimum, median and maximum ESS/L values obtained. We see that the adaptive methods (AHMC and ARMHMC) exhibit good performance. For the minimum ESS/L, AHMC has better (higher) values that NUTS for all the data sets, and this behavior is consistent across most other data sets for the other summary statistics. Thus AHMC typically provides better performance and a higher effective number of samples per unit of computation used than NUTS. We also see that the ARMHMC can improve RMHMC and provide better ESS/L on what is already a highly efficient sampler.

\begin{figure}[t]
\includegraphics[width = \columnwidth]{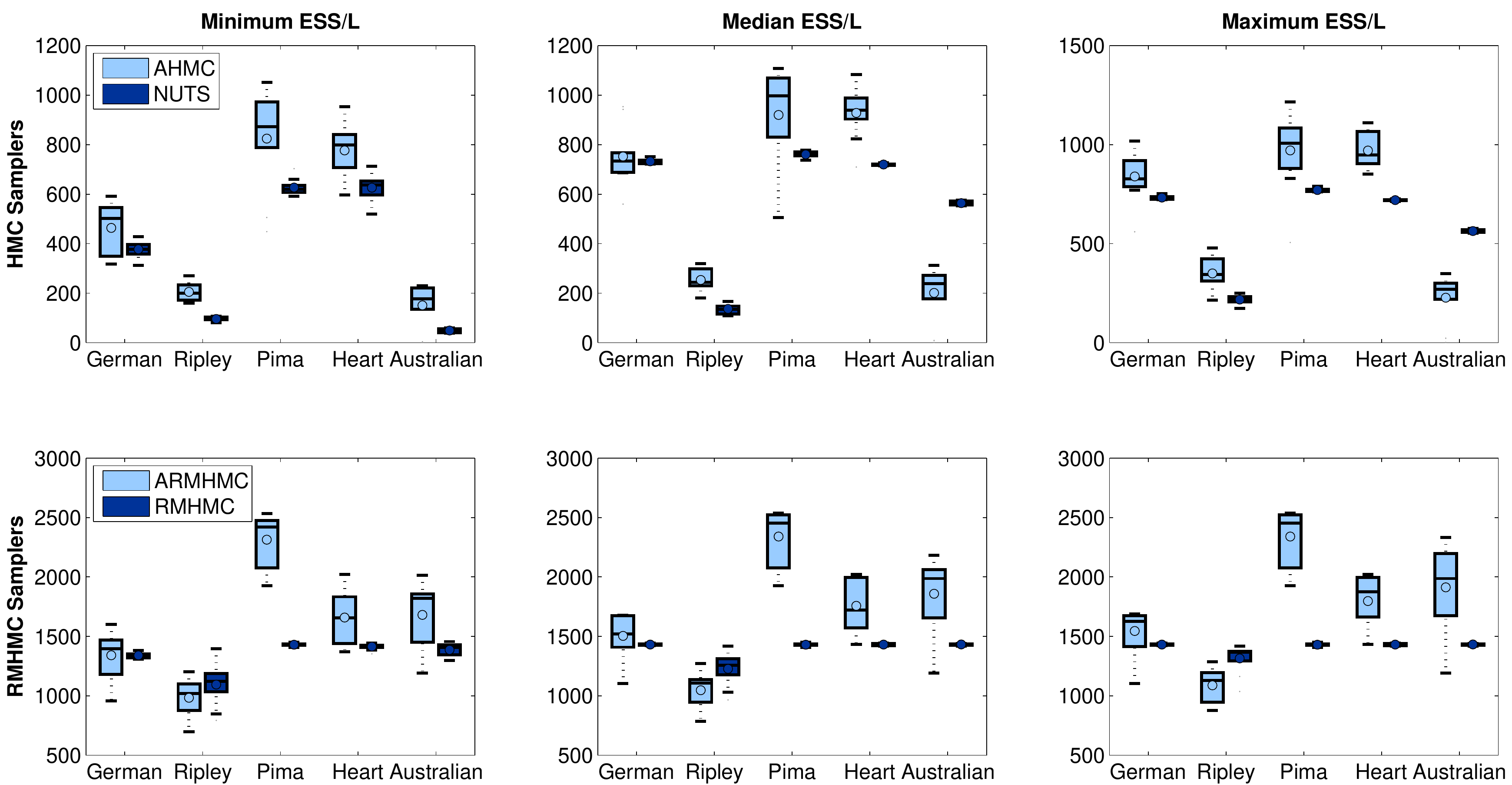}
\caption{Box plots comparing ESS/L for Bayesian logistic regression. Top row: AHMC vs NUTS. Bottom row: ARMHMC vs RMHMC.}
\label{fig:ESSL_BLR}
\end{figure}


\subsection{Log-Gaussian Cox Point Process}
We model a data set $\vY = \{y_{ij}\}$ that consists of counts at locations $(i,j), i,j = 1 \ldots, d$ in a regular spatial grid using a log-Gaussian Cox point process (LGC) \citep{christensen2005scaling, Girolami-11}. Observations $y_{ij}$ are Poisson distributed and conditionally independent given a latent intensity process $\vLambda = \{\lambda_{ij} \}$ with means  $s \lambda_{ij} = s \exp (x_{ij})$, where $s = \frac{1}{d^2}$. The rates $\vX = \{x_{ij}\}$ are obtained from a Gaussian process  with mean function $m(x_{ij}) = \mu \mathbf{1}$ and covariance function $\Sigma(x_{ij}, x_{i' j'}) = \sigma^2 \exp \left( - \delta(i, i', j, j')/\beta d \right)$, where $\delta(i,i', j, j') = \sqrt{(i - i')^2 + (j - j')^2}$. The joint probability $\log p(\vy, \vx | \mu, \sigma, \beta)$ is proportional to:
\begin{align}
 & \sum_{i,j} y_{ij} x_{ij} \! - \! d \exp( x_{ij})\! - \! \frac{1}{2}(\vx \!-\! \mu\mathbf{1})^\top \vSigma^{-1} (\vx \!-\! \mu\mathbf{1}).
\end{align}
We generate samples jointly for $\vx, \sigma, \mu, \beta$ using a grid of size $d = 64$, using a synthetic data set obtained by drawing from the generative process for this model.  We generate 5000 samples after a burnin of 1000 samples.
For this model, we use $L \in \{1, \cdots, 500\}$, $\epsilon \in [0.001, 0.1]$ for AHMC, and use $L \in \{1, \cdots, 60\}$, $\epsilon \in [0.01, 1]$ for ARMHMC.
We compare the performance of the adaptive method we presented in terms of ESS per leapfrog step in figure \ref{fig:ESSL_LGC}. 
We compare AHMC versus NUTS and ARMHMC versus RMHMC, showing the minimum, median and maximum ESS per leapfrog step obtained for 10 chains with dispersed starting points. 
We see that almost all points lie below the diagonal line, which indicates that the AHMC and ARMHMC have better ESS/L compared to NUTS and RMHMC, respectively. Thus even for high-dimensional models with strong correlations our adaptive method allows for automatic tuning of the sampler and consequently the ability to obtain higher quality samples than with competing methods.
\begin{figure}[t]
\centering
\includegraphics[width = 0.45\columnwidth]{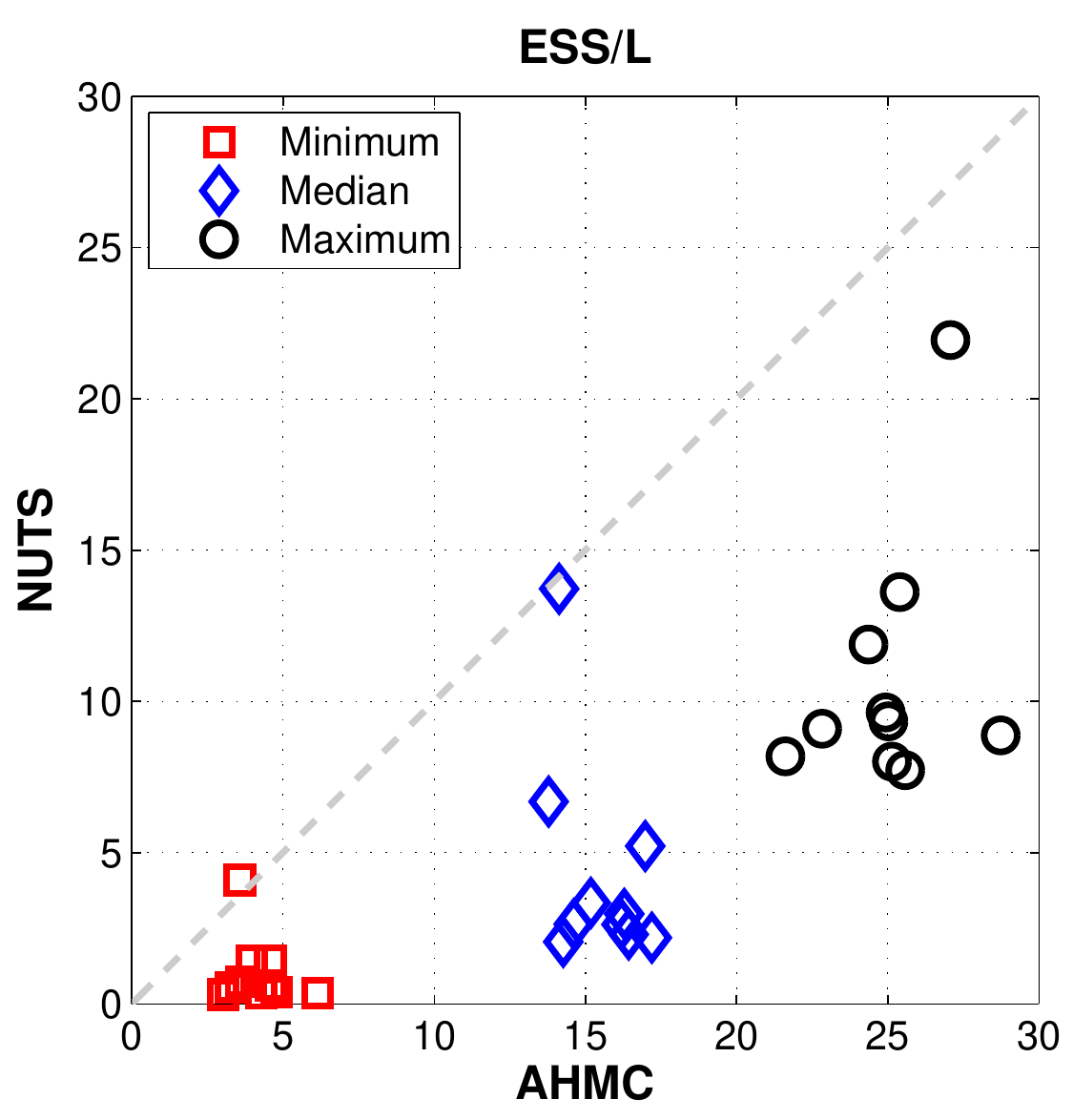}
\includegraphics[width = 0.45\columnwidth]{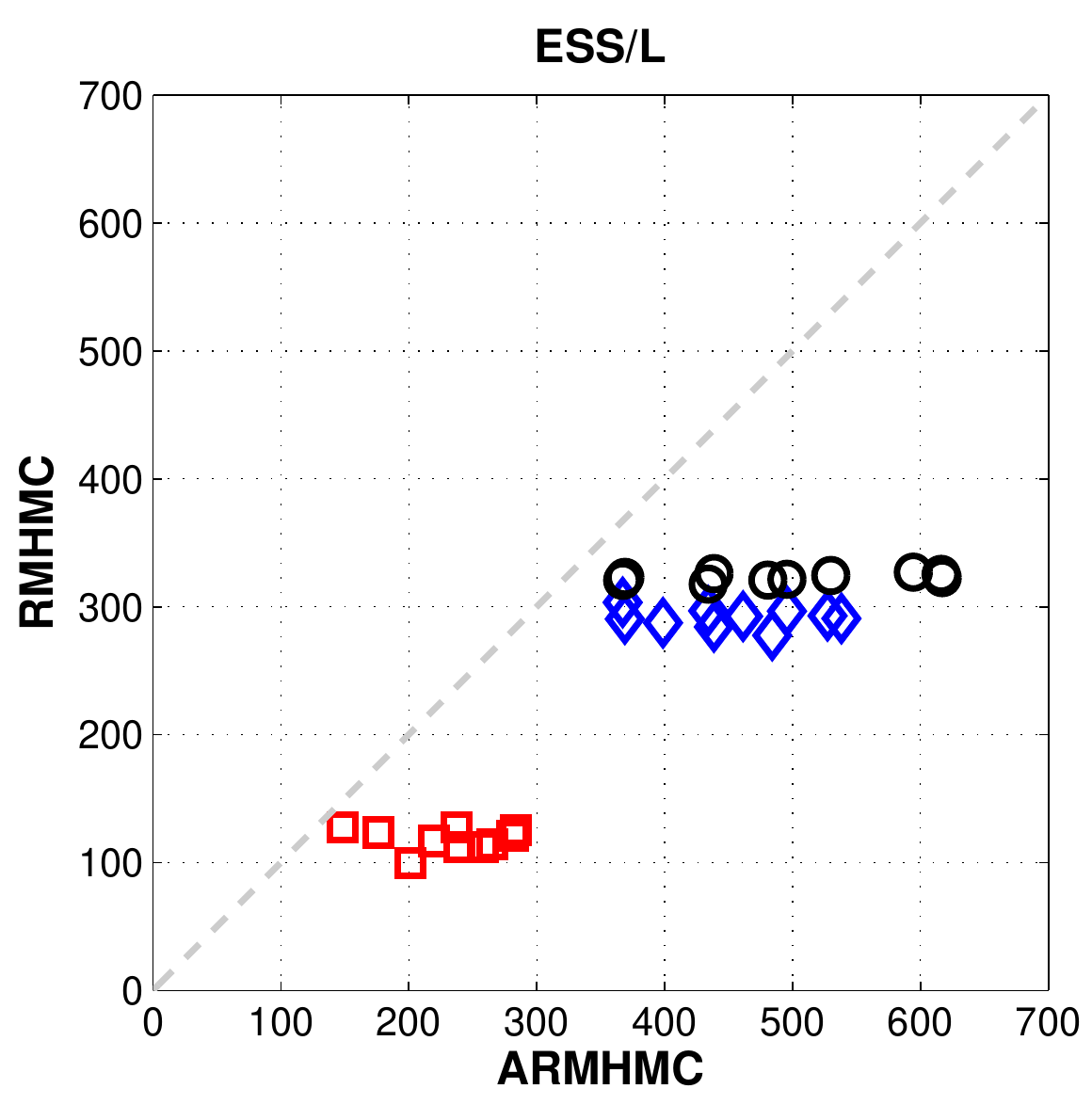}
\caption{Comparing minimum (red), median (blue) and maximum (black) ESS/L for the Log-Gaussian Cox model. Each of the colored glyphs represents one of the 10 chains generated.}
\label{fig:ESSL_LGC}
\end{figure}

We examine the quality of the posterior distribution obtained for AHMC and NUTS in figure \ref{fig:heatmap_LGC}, by visualizing the latent field and its variance, and comparing to the true data (which is known for this data set). The top row shows the true latent fields. From the true data observations (shown in top right corner), we see that there are few data points in this region and thus we expect to have a high variance in this region. The average of the samples obtained using AHMC shows that we can accurately obtain samples from the latent field $\mathbf{x}$, and that the samples have a variance matching our expectations. While NUTS is able to also produce good samples of the latent field, the variance of the field is not well captured (bottom right image).
\begin{figure}[t]
\centering
\includegraphics[height = 7cm]{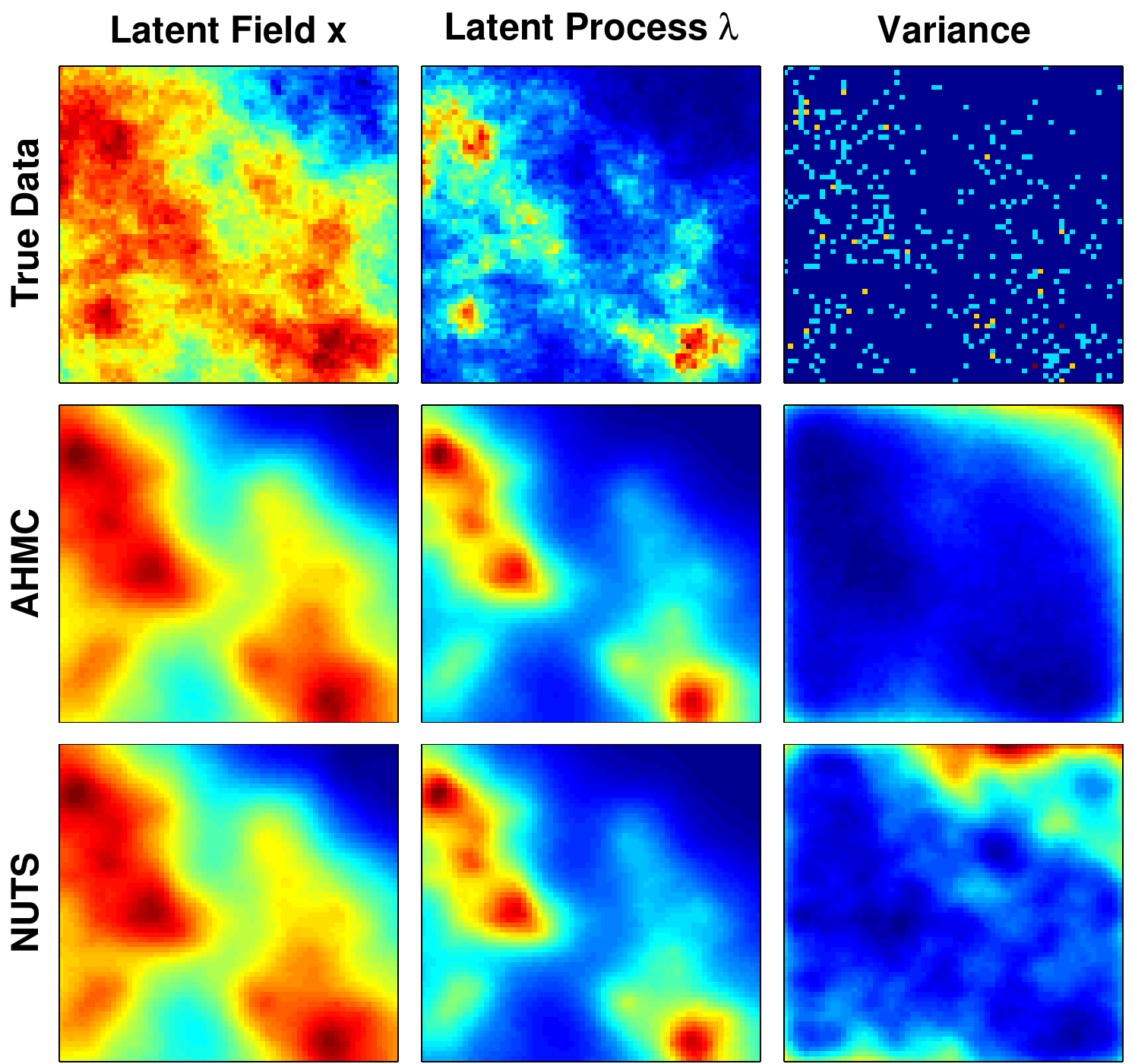}
\caption{Comparing quality of posterior distributions from samples obtained using AHMC and NUTS for the log-Gaussian Cox model. The top-right image shows the locations of the true data.}
\label{fig:heatmap_LGC}
\end{figure}

\subsection{Stochastic Volatility}
We consider a stochastic volatility model described by \citet{Kim-98} and \citet{Girolami-11}, in which we consider observations $y_t$, regularly spaced in time for $t = 1, \ldots, T$. Each $y_t$ is specified using a latent variable $x_t$, which represents the log-volatility following auto-regressive AR(1) dynamics. The model is specified as:
\begin{align}
y_t = & \epsilon_t \beta \exp \left(0.5 x_t\right), & \epsilon_t \sim & \mathcal{N}(0,1)\\
x_{t+1} = & \phi x_t + \eta_{t+1}, & \eta_{t+1} \sim & \mathcal{N}(0, \sigma^2) \\
x_1 \sim & \mathcal{N} \left( 0, \frac{\sigma^2}{1 - \phi^2}\right), & p(\beta) \propto & \frac{1}{\beta}.
\end{align}
For stationarity of the log-volatility, $|\phi| <1$, and the standard deviation $\sigma > 0$, whose priors we set to $\tfrac{\phi + 1}{2} \sim Beta(20,1.5)$ and $\sigma^2 \sim\textrm{inv-}\chi^2(10, 0.05)$, respectively.
The parameters to be sampled by HMC is thus $\boldsymbol{\Theta} = \{\mathbf{x}, \beta, \phi, \sigma^2 \}$, and the joint probability is:
\begin{align}
p(\mathbf{y},\! \boldsymbol{\Theta})\!\! =\!\!  \prod_{t=1}^T \! p(y_t | x_t, \beta) p(\!x_t | x_{t\!-\!1}, \phi, \sigma^2\!) p(\!\beta\!) p(\!\sigma^2\!)p(\!\phi\!).
\end{align}
 We make use of the transformations $\sigma = \exp(\gamma)$ and $\phi = \tanh(\alpha)$ to ensure that we sample using unconstrained variables; the use of this transformation requires the addition of the Jacobian of the transformation of variables. We generate samples jointly using our AHMC methods, using training data with $T = 2000$. 
 For our experiments, we use a burnin period of $10,000$ samples and thereafter generate $20,000$ posterior samples. 
 We restrict our box constraint such that $L \in \{1, \cdots, 300\}$, $\epsilon \in [10^{-4}, 10^{-2}]$.
 We show the results comparing ESS for the two methods in table \ref{tab:SV}. These results again show higher values for ESS per leapfrog step, demonstrating that a better performing sampler can be obtained using AHMC -- further demonstrating the advantages of AHMC methods for sampling from complex hierarchical models.

\begin{table}[t]
    \scriptsize
    \centering
    \caption{Comparative results for the stochastic volatility model. }
    \begin{tabular}{| c | c | c| c|}
        \hline
         &  \multicolumn{3}{c|}{\textbf{ESS per Leapfrog}} \\
        \hline
        Sampler & minimum & median & maximum\\ 
        \hline
        \textbf{AHMC} & 1.3 $\pm$ 0.1 & 6.9 $\pm$ 0.7 & 14.9 $\pm$1.4\\ 
        \textbf{NUTS} & 0.7 $\pm$ 0.3 & 3.5 $\pm$1.6 & 9 $\pm$2.8 \\
        \hline
    \end{tabular}
    \label{tab:SV}
    \vspace{-0.5cm}
\end{table}

\subsection{Bayesian Neural Networks}
\label{sect:expt_BNN}
We demonstrate the application of our adaptive approach using 
Bayesian neural networks (BNNs) to show that AHMC allows for 
more effective sampling of posterior parameters even when compared to samplers 
finely tuned by an expert. We make use of the Dexter data set from the NIPS 2003 
feature selection challenge, which is a subset of the well-known Reuters text categorization
benchmark.
The winning entries submitted by 
\citet{ Neal-06} used a number of feature selection 
techniques followed by a combination of Bayesian Neural Networks and Dirichlet diffusion trees. 
The entry that used only BNNs was placed second and achieved highly competitive results~\cite{Guyon-05}.

The BNN model consists of 295 input features and 2 hidden layers 
with 20 and 8 hidden units respectively. The input features are selected 
from the full set of features through univariate feature selection. 
The weights and bias as well as a few other parameters of this particular 
network adds up to form a 6097 dimensional state space for the HMC sampler. 

For this model, we use cross-validation to construct the reward signal.
We divide the data into $n$ sets, 
and train $n$ BNNs
each on $n - 1$ sets and test them on the remaining set like in the case of 
normal cross-validation. 
The cross-validation error is then used to calculate the reward.
To take computation into account, we always evaluate the reward over the same 
number of leapfrog steps, i.e. for each evaluation of the reward we use a 
different number of samples and a different number of leapfrog steps for each
sample, but the product of the two remains constant.

We compare the results in table \ref{tab:dextertable_test}, 
where the performance measure is the prediction error on a test set 
(unknown to us) and was obtained after submission to the competition system. 
The improved results obtained using the AHMC strategy are clear from the table, 
also demonstrating that good adaptation can be preferable to the introduction 
of more sophisticated models.

\begin{table}[h]
  \scriptsize
  \caption{Classification error on the test set of the Dexter data set.
      The table shows the mean and the median prediction errors
      of our $8$ BNNs trained as in cross-validation. The majority vote
      of these $8$ networks achieves slightly better results 
      than that of a more sophisticated
      model involving Dirichlet diffusion trees.}
  \label{tab:dextertable_test}
    \begin{center}
      \begin{tabular}{lr}
      \multicolumn{1}{l}{\bf Method}  &\multicolumn{1}{c}{\bf Error} \\
      \hline 
      Expert-tuned HMC for BNN                &0.0510\\
      AHMC for BNN (Mean error)       &0.0498\\
      AHMC for BNN (Median error)     &0.0458\\
      \hline
      Winning entry (using Dirichlet Diffusion Trees)  &0.0390\\
      AHMC for BNN + Majority Voting  &0.0355 \\
      \hline
      \end{tabular}
    \end{center}
    \vspace{-0.7cm}
\end{table}

\section{Discussion and Conclusion}
\label{sec:concl}
In section \ref{sec:ahmc} we described the use  the expected squared jumping distance as a suitable objective. Several other objectives, such as the mean update distance, cross-validation error and the cumulative auto-correlation, are also suitable, and their use depends on the particular modelling problem. In many machine learning tasks, researchers design MCMC algorithms to estimate model parameters and, subsequently, evaluate these models using cross-validation, such as the competition task in section \ref{sect:expt_BNN}. 
Moreover, often researchers modify their samplers so as to reduce test set error. 
In this paper, we demonstrate the use of predictive losses, 
such as cross-validation error, to guide the adaptation. 
This approach, although never reported before to the best of our knowledge, 
simply makes the tuning process followed by many researchers explicit.
Ultimately the models whose parameters we are 
estimating by running a Markov chain will be tested on predictive tasks. 
Hence, it is natural to use predictive performance 
on such predictive tasks to improve 
the exploration of the posterior distribution. 
Of course, these objective measures are only applicable when 
sufficient data is available to obtain good predictive estimates.
\\ \\
We addressed the widely-experienced difficulty in tuning Hamiltonian-based Monte Carlo samplers by developing algorithms for infinite adaptation of these Markov chains using Bayesian optimization. The adaptive Hamiltonian Monte Carlo and adaptive Riemann manifold HMC we developed automate the process of finding the best parameters that control the performance of the sampler, removing the need for time-consuming and expert-driven tuning of these samplers. Our experiments show conclusively that over a wide range of models and data sets, the use of adaptive algorithms makes it easy to obtain more efficient samplers, in some cases precluding the need for more complex approaches. 
Hamiltonian-based Monte Carlo samplers are widely known to be an excellent choice of MCMC method, and we hope that this paper removes a key obstacle towards the more widespread use of these samplers in practice.

{
\small
\bibliography{ziyu}
\bibliographystyle{icml2013}
}

\end{document}